\documentclass[makeidx]{amsart}
\usepackage{epsfig}
\usepackage{amsmath}
\usepackage{amssymb}
\newtheorem{Proposition}{Proposition}
  \newtheorem{Remark}{Remark}
  \newtheorem{Corollary}[Proposition]{Corollary}
  \newtheorem{Lemma}[Proposition]{Lemma}
  
  \newtheorem{Theorem}{Theorem}

 \newtheorem{Note}[Remark]{Note}

\newcommand {\z}{{\noindent}}
 
\def\CC{\mathbb{C}}
 \def\RR{\mathbb{R}}
 \def\NN{\mathbb{N}}

\def\Re{\mathrm{Re}}
\def\Im{\mathrm{Im}}
\def\ds{\displaystyle}
 \def\({\left(} \def\){\right)} \makeindex
\author{O. Costin, M. Huang} \address{Mathematics Department\\The Ohio State University\\Columbus, OH 43220} \title[Gamow vectors and Borel summability]{Gamow Vectors and   Borel summation}
\begin{document}
\begin{abstract}
  We analyze
 the detailed time dependence of the wave function
  $\psi(x,t)$ for one dimensional Hamiltonians $H=-\partial_x^2+V(x)$
  where $V$ (for example modeling barriers or wells) and $\psi(x,0)$
  are {\em compactly supported}.

  We show that the dispersive part of $\psi(x,t)$, its asymptotic
  series in powers of $t^{-1/2}$, is Borel summable.  The remainder,
  the difference between $\psi$ and the Borel sum, is a convergent
  expansion of the form $\sum_{k=0}^{\infty}g_k
  \Gamma_k(x)e^{-\gamma_k t}$, where $\Gamma_k$ are the Gamow vectors
  of $H$, and $\gamma_k$ are the associated resonances; generically,
  all $g_k$ are nonzero. For large $k$, $\gamma_{k}\sim const\cdot
  k\log k +k^2\pi^{2}i/4$. The effect of the Gamow vectors is visible
  when time is not very large, and the decomposition defines
  rigorously resonances and Gamow vectors in a nonperturbative regime,
 in a physically relevant way.

  The decomposition allows for calculating
  $\psi$  for moderate and large $t$, to any prescribed
  exponential accuracy, using optimal truncation of power series plus finitely many Gamow vectors contributions.

  The analytic structure of $\psi$ is perhaps surprising: in general
  (even in simple examples such as square wells), $\psi(x,t)$ turns
  out to be $C^\infty$ in $t$ but nowhere analytic on $\RR^+$. In
  fact, $\psi$ is $t-$analytic in a sector in the lower  half plane and has
 the whole of $\RR^+$ a natural boundary.

 Extension to other types of potentials, for instance
  analytic at infinity, is briefly discussed, and in the process we
  study the singularity structure of the Green's function in a
  neighborhood of zero, in energy space.

\end{abstract}
\maketitle

\tableofcontents
\section{Introduction}

\label{first}

\z Resonances play a major role in the physics of metastable states
and their decay.  From a mathematical standpoint, there is a good
number of definitions of resonances and resonant states. In most
approaches, they are based on the properties of the scattering matrix,
on Gelfand triples (rigged Hilbert spaces), or on the complex analytic
singular structure of the Green's function beyond the spectrum of the
resolvent.  The pole positions of the Green's function,
``resonances''  are pseudo-eigenvalues, and their residues (Gamow vectors)
are pseudo-eigenvectors of the Hamiltonian with ``purely growing''
conditions at infinity. There is a vast literature on the subject, see
e.g. the concise overview \cite{Madrid1} and the references
therein. See also \cite{Tumulka} for a surprising consequence of
resonant states, and for a clear description of the physical relevance
of Gamow vectors.

By and large, the different mathematical definitions provide
equivalent objects. However, there are conceptual difficulties in all
rigorous  approaches, and these lie in connecting (a) the mathematical
definition, (b) the natural properties of the underlying quantum
Hamiltonian, and (c) the physical phenomenon. In fact, {\em Howland's
  Razor}, a principle so dubbed by B. Simon, cf. \cite{Simon1}, states
that {\em no satisfactory definition of resonance can depend on the
  structure of a single operator on an abstract Hilbert
  space}. Slightly oversimplifying the argument, the reason is that
the analytic  structure of the  Green's
function, or of quantities obtained through dilation-analyticity, needed
in most approaches, are by
no means unitarily invariant.  Unitary invariance plays of course an
important role in quantum mechanics since observables are represented
by self-adjoint operators on Hilbert spaces, the isomorphisms of which
are precisely the family of all unitary transformations.

 A concise and  very illuminating  critical
analysis of the various mathematical attempts at rigorous
definitions is found in \cite{Simon1}.

\subsection{Resonances and asymptotic expansions}
We note however that many relevant physical quantities are not and
need not be defined in a unitarily-equivalent way. As already
mentioned, resonances are used in measuring the time decay of the
probability distribution in physical space. In any interpretation of
quantum mechanics, $L^2(\RR^3)$ plays a distinguished role, when
$\RR^3$ is a representation of the space where we, and macroscopic
apparatuses, lie.

A definition based on time behavior is natural to the underlying
physics and  avoids Howland's razor since it rests on (i) a
particular representation of the Hamiltonian--as an operator on $L^2$ of our
$\RR^3$, (ii) on a second observable, say $\mathbf{1}_{A}$, the
characteristic function of the set $A\in\RR^3$ and (iii) on a specific
mathematical question--the time decay of $\langle\psi
|\mathbf{1}_{A}|\psi\rangle$. This triad is not (at least not
manifestly) a property of a single operator. Nonetheless,
$L^2(\RR^3)$ and dependence on time are canonical objects in
analyzing scattering or decay problems.

At the present time however rigorous definitions based on  time
behavior only exist in a perturbative regime, \cite{Skibsted},
\cite{Madrid2}; see also below.

In this paper, for compactly supported
potentials in one dimension, we show that the difference
between the wave function and the Borel sum of its asymptotic series
in powers of $t^{-1/2}$ is a convergent expansion in Gamow
vectors. The resonances thus defined turn out to be independent of the
initial condition. Gamow vectors are not $L^2$ functions; neither is
the Borel sum (see \S\ref{transth}) of the power series. The expansion
is valid uniformly on compact sets instead.

The representation as a Borel summed series plus Gamow vectors
expansion is valid not only for large $t$, but, in fact,  simply
for $t>0$, though the convergence rate of the whole expansion is
rapid enough only  if  $t$ is not too small.

After completing a manuscript we found that
decompositions in energy space in terms of Gamow vectors
and a continuous part have been proposed in the physics literature,
see \cite{Calderon}, to our knowledge without completely rigorous, mathematical,  proofs or
study of Borel summability, and with a different interpretation and
suggested physical meaning; cf.   Note \ref{N6} below.  Without
Borel summability, uniqueness of a decomposition in terms of a
continuum integral and a sum of exponentials generally does not
hold, see \S\ref{transth}.

The $t^{-1/2}$ power series expansion roughly corresponds to the decay
of a free particle \footnote{The influence of the potential --however
  distant-- is still present in the ``initial state'' at some very late
  time $t_i\gg 1$, {\em from which} the almost free particle decays;
  the state at $t_i$ is responsible for the generic disappearance of the
zero
  energy resonance.}. Indeed, if time is very long and the point
spectrum of $H$ is empty, then, eventually, the overlap between the
wave function and the support of the potential becomes negligible. The
specifics of the potential are seen while the particle has a fair
probability of it being near the potential. This is why it is natural to
subtract out the power series, ``free'' decay. But, generally, the series has zero
radius of convergence\footnote{If the potential is
  unbounded, such as a dipole $V(x)=Ex$, then the power series may be
  identically zero, see \cite{Herbst}, \cite{Rokhlenko} and references
  therein. Another exception is $V=0$, for which the asymptotic $t^{-1/2}$ series
  converges on compact sets in $x$.}.

If parameters are such that a resonance (complex generalized
eigenvalue, \cite{Madrid2}) is at a small distance $\epsilon$ to the
spectrum of $H$, the setting is called perturbative and there is a
time scale, roughly given by $e^{-\epsilon t}\gg t^{-3/2}$, during
which in a finite spatial interval, the decay of the position
probability follows an exponential law. This corresponds to a
transient, metastable state. The Gamow vector corresponding to such a
resonance describes the wave function on increasingly larger spatial
regions, see \cite{Tumulka}, \S 9. Only metastable states with long
enough survival time are captured however in this way. (Rigorously
speaking, we are dealing with a double limit, in which time goes to
infinity and an external parameter goes to zero in some correlated
fashion.)  Borel summation provides an exact representation for all
$t>0$, as well as practical ways to calculate the wave function for
times of order one, see \S\ref{numcalc}; the influence of resonances
which are not necessarily close to the spectrum is measurable.

For showing Borel summability, perhaps the most delicate part is the
analysis of the Green's function in the fourth quadrant in the energy
parameter, where infinitely many poles recede rapidly to infinity;
sharp estimates are needed in order to control a needed Bromwich
contour integral.

Extension to other potentials with sufficient analyticity and decay is discussed
in \S\ref{Ap}.

\section{Setting and main results}
\z We consider the one-dimensional Schr\"odinger
equation \[ i\hbar\dfrac{\partial}{\partial
  t}\psi(x,t)=-\dfrac{\hbar}{2m}\dfrac{\partial^{2}}{\partial
  x^{2}}\psi(x,t)+V(x)\psi(x,t)\] where:

(a) The nonzero potential $V$ is independent of time, compactly
supported and $C^2$ on its support. ($V$ is allowed to be
discontinuous at the endpoints provided that it is one-sided $C^2$
at the endpoints.)

(b) The initial condition $\psi_0(x)$ is compactly supported and
$C^2$ on its support.

 We normalize the equation to
\begin{equation}
  i\dfrac{\partial}{\partial
    t}\psi(x,t)=-\dfrac{\partial^{2}}{\partial
    x^{2}}\psi(x,t)+V(x)\psi(x,t)=(H\psi)(x,t)\label{eq:ori}\end{equation}
where
$\textrm{supp}(V)\subset[-1,1]$.
Under the assumptions above, we have the following results.
\begin{Proposition} For large $t$, the wave function
  $\psi(x,t)$ is $O(t^{-1/2})$ (in the generic case of absence
of zero energy resonance \cite{Goldberg}, $\psi(x,t)=O(t^{-3/2})$), and  $\psi(x,t)$ has a Borel summable
asymptotic series $\tilde{\psi}(x,t)$ in powers of $t^{-1/2}$.
\end{Proposition}
We denote as usual by $\mathcal{LB}$ the Borel summation operator.
Let $t^{-1/2}\varphi(x,t)=\mathcal{LB}\tilde{\psi}(x,t)$ where
$\varphi(x,\cdot)$ is bounded. As seen below,
$\psi(x,t)-t^{-1/2}\varphi(x,t)$ is nonzero, and is a convergent
combination of Gamow vectors, the residues at the poles of the
analytic continuation of the resolvent of $H$.

 Let $\{E_k\}_{k=1,...,N}$ be the eigenvalues of $H$
and  $\{\psi_k\}_{k=1,...,N}$  be the corresponding 
eigenfunctions.
(We convene to set  $N=0$ if these two sets are empty.) Let also
$\gamma_k, \Re\gamma_k>0$ be  the generalized
eigenvalues (resonances) corresponding to the Gamow vectors $\Gamma_{k}(x)$.

\begin{Theorem} \label{T1}
(i) For all  $t>0$ we have

\begin{equation}
  \psi(x,t)-t^{-1/2}\varphi(x,t)=\sum_{k=1}^{N}b_k\psi_{k}(x)e^{-E_{k}it}+\sum_{k=1}^{\infty}g_k\Gamma_{k}(x)e^{-\gamma_{k}t}\label{eq:dec}\end{equation}
The infinite sum in (\ref{eq:dec}) is uniformly convergent on compact sets in $x$ --rapidly
so if $t$ is large.
  (The coefficients $b_k$ and $g_k$ depend on
$\psi$  and
typically  $g_k\ne 0$ for all $k$.)

\smallskip

(ii) $\psi_{k}(x),\Gamma_{k}(x),\varphi(x,t)$ are twice
differentiable in $x$.

\smallskip

(iii) We have
\begin{equation}
  \label{eq:lmbd}
  \gamma_{k}\sim
 const\cdot k\log k +k^2\pi^{2}i/4
\text{ as }k\rightarrow +\infty
\end{equation}
(Higher orders depend on $V$, see Proposition \ref{asy}.)
The $\gamma_k$ are independent of $\psi_0$, and  the  constant
depends on the endpoint behavior of $V$.
\end{Theorem}
The series in (\ref{eq:dec}), though valid for all $t$, converges
poorly if $t\to 0$: this is not the regime it is intended for.

Let

\begin{equation}\label{defE}
  \mathcal{E}(u,t)=\sqrt{\frac{x}{t}}+e^{-u^2 t} \left( u^2 \sqrt{\pi t} E_{\frac{1}{2}}(-u^2 t)+u E_1(-u^2 t)\right)
\end{equation}
where {\rm E}$_n$ is the $n$-exponential integral and $\arg u\in
(-\pi,0)\cup (0,\pi)$, \cite{Abramowitz} \footnote{There seem to be inconsistent definitions in the
literature. We use \cite{Abramowitz};  since the definition is not
spelled out in one place, we state it again: $\mathrm
E_n(z)=\int_1^{\infty}t^{-n} e^{-zt}dt$ for $\Re \, z >0$,
analytically   continued to $\CC\setminus \RR^-$, and extended by continuity to the two sides of the cut.}

\begin{Proposition}[$o(e^{-Mt})$ accuracy, for arbitrary $M$]\label{approx2} For any $M$ there exists (explicit) $m$ and $m_1\le m$,
so that
  \begin{equation}
    \label{eq:eqapprox}
     \psi(x,t)= \sum_{k=1}^{N}b_k\psi_{k}(x)e^{-E_{k}it}+\sum_{k=1}^{m_1}g_k\Gamma_{k}(x)e^{-\gamma_{k}t}-\sum_{k=1}^m r_k \mathcal{E}(\tilde{\gamma}_kt)+\psi_M(x,t)
  \end{equation}
  Here $\{\tilde{\gamma}_k\}_{k\le m}$ are the poles of the Green's
  function (resonances) on the first and second Riemann sheet with $|
  \tilde{\gamma}_k|\le M$, $\{\gamma_k\}_{k\le m_1}$ are the subset of them on the
  first Riemann sheet, $r_k$ are the corresponding residues, and
  $\psi_M$  differs  by $o(e^{-Mt})$ from its (relatively explicit) power
  series in $1/t^{1/2}$, optimally truncated (see \S\ref{numcalc}).
\end{Proposition}
\begin{Note}{\rm
(i) We see that, as exponential contributions, only the resonances
on the first Riemann sheet appear but both sheets contribute
to the dispersive part.

(ii) The expression (\ref{defE}) is not analytic on the Riemann
surface of the log: it has a jump on $\RR^+$, compensated by an
opposite jump of $\psi_M$. These jumps are mandated by
least term summability requirements.}

\end{Note}
\begin{Corollary} {\em Any number of resonances can be
    calculated from $\psi(x,t)$, if $\psi$ is known with
    correspondingly high accuracy. Conversely, $\psi$ can be
    calculated in principle with arbitrary accuracy from the
    contribution of a finite number of bound states, resonances,
    exponential integrals and optimal truncation of power series.}
\end{Corollary}
\begin{Note}[Analytic structure of $\psi(x,t),\text{ in
  }t\in\RR^+$]{\rm It follows from the proof that $\ln |g_k\Gamma_k|=O(\sqrt{p_k})$. Thus, since $\varphi$ is manifestly analytic
    for $\Re\,\, t>0$, it follows immediately from (\ref{eq:dec}) that
    $\psi$ is $C^\infty$ in $t$. Now, since $\Im\,\,-\gamma_k\sim
    -k^2$, near $\RR^+$, $\psi$ equals a function analytic in the right
    half plane (the Laplace transform) plus a lacunary Dirichlet
    series, convergent for $\Im\,\, t<0$ (the ``heat-like''
    direction).  For generic $x$, the coefficients of the Dirichlet
    series are bounded below by $e^{-\text{const}(x)\sqrt{|p|}}$ (all
    functions involved are of exponential order $1/2$; the lower
    bounds follow relatively easily from the proofs, but we omit the
    details). Then the Dirichlet series does not converge past
    $\RR^+$; general theorems on lacunary series, see
    e.g. \cite{Mandelbrojt} imply then that $\RR^+$ is a natural
    boundary. See also Proposition \ref{P8}. Using similar estimates
    it can be checked that the Taylor coefficients of $\psi$ at a
    point $t_0$ behave roughly like $\displaystyle
    \left(\frac{\pi^4k^2}{4e^2 t_0^2\ln^2 k}\right)^k$, showing once
    more that there is no point of analyticity on $\RR^+$.  This is
    another way to see the contribution of the Gamow vectors to the
    properties of $\psi$. (More details about this are part of a future
    paper.)}

\end{Note}
\section{Proofs of Main Results}

\subsection{Integral reformulation of the problem}\label{main}
$H$ satisfies the
     assumptions of Theorem X.71, \cite{Reed-Simon} v.2 pp 290. Thus,
for any $t$, $\psi(t,\cdot)$ is in the domain of $-d^2/dx^2$. This implies
     continuity in $x$ of $\psi(t,x)$ and of its $t-$Laplace
     transform. It also follows that the unitary propagator $U(t)$ is
     strongly differentiable in $t$.
Existence of a strongly differentiable unitary propagator for
(\ref{eq:ori}) implies  existence of  the Laplace
transform
$$\hat{\psi}(x,p)=  \int_0^{\infty}e^{-pt}\psi(x,t)dt=\left(\int_0^{\infty}e^{-pt}U(t)dt\right)\psi_0(x)$$
for $\Re \,p>0$. Taking the Laplace transform of (\ref{eq:ori}) we obtain
\begin{equation}
ip\hat{\psi}(x,p)-i\psi_{0}(x)=-\dfrac{\partial^{2}}
  {\partial
    x^{2}}\hat{\psi}(x,p)+V(x)\hat{\psi}(x,p)\label{eq:lap0}
\end{equation}
where $\psi_{0}(x)$ is the initial condition. Treating $p$ as a parameter, we
write $\psi(x,p)=y(x;p)=:y(x)$, and obtain
\begin{equation}
  y''(x)-\left(V(x)-ip\right)y(x)=i\psi_{0}(x)\label{eq:lap}
\end{equation}
where $y(x)\in L^2(\RR)$. The associated homogeneous equation is
\begin{equation}
  y''(x)=\left(V(x)-ip\right)y(x)\label{eq:hom}
\end{equation}
If $y_{+}(x),y_{-}(x)$ are two linearly independent solutions of
(\ref{eq:hom}) with the additional restrictions (and the usual branch of the log)
\begin{eqnarray}
  \label{eq:defypm}
  y_{+}(x)=e^{-\sqrt{-ip}x}\:\mathrm{when\:}x>1\nonumber \\
y_{-}(x)=e^{\sqrt{-ip}x}\:\mathrm{when\:}x<-1
\end{eqnarray}
then, for $\Re\,p>0$, the
$L^2$  solution of (\ref{eq:lap0}) (or equivalently of 
(\ref{eq:lap})) is
\begin{equation}
  \hat{\psi}(x,p)=\frac{i}{W_{p}}\left(y_{-}(x)\int_{+\infty}^{x}y_{+}(s)\psi_{0}(s)ds-y_{+}(x)\int_{-\infty}^{x}y_{-}(s)\psi_{0}(s)ds\right)\label{eq:ps1}
\end{equation} where the Wronskian
$W_{p}=y_{+}(x)y_{-}'(x)-y_{-}(x)y_{+}'(x)$ is easily seen to be
independent of $x$.

As we shall see, this solution is meromorphic in $p$ except for a
possible branch point at 0, and for fixed $x$ it has
sub-exponential bounds in the left half $p$-plane (when not close to
poles). The function $\psi$ is the inverse Laplace transform of
$\hat{\psi}$, and it can be written in the form
$\psi(x,t)=\frac{1}{2\pi
  i}\int_{a_0-i\infty}^{a_0+i\infty}\hat{\psi}(x,p)e^{pt}dp$. We show that
the contour of integration can be pushed through the left half plane;
collecting the  contributions from poles and branch points, the
decomposition follows.
\begin{Note}{\rm The domain of interest in $p$ is a sector on the Riemann
surface of the square root, centered on $\RR^+$ and of opening slightly
more than $2\pi$, which, in the variable $\sqrt{-ip}$ translates into a sector of opening
more that $\pi$ centered at $\sqrt{-i}$.
}\end{Note}

\subsection{Analyticity of $\hat{\psi}$ on the Riemann surface of $\sqrt{p}$}
We start with the  analyticity properties of $\hat{\psi}$.
The more delicate analysis of the asymptotic behavior
of the analytic continuation of $\hat{\psi}$ on the Riemann
surface of the log at zero is done in \S\ref{ilt1}. The
existence of a square root branch point at zero is typical
in this type of problems. For our analysis, in proving Borel
summability, we need to show that $\hat{\psi}$ is meromorphic in $\sqrt{p}$,

\begin{Proposition} $\hat{\psi}(x,p)$ is meromorphic in $p$ on the Riemann surface of the square root at zero, $\mathbb{C}_{1/2;0}$ and zero is
  a possible square root branch point.
\end{Proposition}
\begin{proof}
This
follows from the following simple argument. Note first that
continuity of $y$ and $y'$ imply  the following matching conditions:
   $$\begin{cases}
    y_{+}(1)=e^{-\sqrt{-ip}}\\
    y_{+}'(1)=-\sqrt{-ip}e^{-\sqrt{-ip}}\\
y_{-}(-1)=e^{\sqrt{-ip}}\\
    y_{-}'(-1)=\sqrt{-ip}e^{\sqrt{-ip}}
   \end{cases}$$
Consider now the solutions $f_1$ and $f_2$ of (\ref{eq:hom}) with
{\em initial conditions }  $f_1(-1)=1$, $f'_1(-1)=0$ and
 $f_2(-1)=0$, $f'_2(-1)=1$. By standard  results on analytic
 parametric-dependence of solutions of differential equations (see, e.g. \cite{Hille}), we see that $f_1$ and $f_2$ are defined on $\RR$
 and for fixed $x$ they are entire in $p$.
We note that, by construction, the Wronskian $[f_1,f_2]$ is one. Then,
$$y_{+}(x)=C_{1}f_{1}(x)+C_{2}f_{2}(x),\,
  y_{-}(x)=C_{3}f_{1}(x)+C_{4}f_{2}(x)$$
where
\[
C_{1}=\sqrt{-ip}e^{-\sqrt{-ip}}\left(f_{2}(1)-f_{2}'(1)\right)\]
\[
C_{2}=-\sqrt{-ip}e^{-\sqrt{-ip}}\left(f_{1}(1)-f_{1}'(1)\right)\]
\[
C_{3}=-\sqrt{-ip}e^{-\sqrt{-ip}}\]
\[
C_{4}=\sqrt{-ip}e^{-\sqrt{-ip}}\]
Furthermore,
\begin{equation}
  W_{p}=-e^{-2\sqrt{-ip}}
\bigg(ip(f_{2}(1))+f_{1}'(1))
 -\sqrt{-ip}(f_{1}(1)+f_{2}'(1))\bigg)\label{eq:wp}
\end{equation}\label{lemma1}
\z Thus $y_{\pm}$ and $W_p$ are analytic in $\mathbb{C}_{1/2;0}$ with a possible
branch point at zero. The same follows for $\hat{\psi}$, by inspection, if
we rewrite its expression as
  \begin{multline}
    \hat{\psi}(x,p)=\frac{i}{W_{p}}\left(y_{-}(x)\bigg(\int_{1}^{x}y_{+}(s)\psi_{0}(s)ds+\int_{+\infty}^{1}e^{-\sqrt{-ip}s}\psi_{0}(s)ds\right)\\-y_{+}(x)\left(\int_{-1}^{x}y_{-}(s)\psi_{0}(s)ds+\int_{-\infty}^{-1}e^{\sqrt{-ip}s}\psi_{0}(s)ds\right)\bigg)\label{eq:psi}
  \end{multline}
\end{proof}

\subsection{The poles for large $p$ in the left half plane}

To effectively calculate the asymptotic position of poles as
$p\rightarrow\infty$ in the left half plane, we need a more
convenient choice for $f_{1},f_{2}$. In the previous subsection they
were chosen to be analytic in $p$. Here we choose a new pair of
$f_{1},f_{2}$ for which the asymptotic behavior as
$p\rightarrow\infty$ is manifest.
\begin{Note} {\rm It is straightforward to check that if $f_{1}(x)$ and $f_{2}(x)$ are solutions of
  (\ref{eq:hom}), and their Wronskian $W_{f;p}=[f_1,f_2]$
is nonzero, then in the decomposition
  $y_{+}(x)=C_{1}f_{1}(x)+C_{2}f_{2}(x),\,
  y_{-}(x)=C_{3}f_{1}(x)+C_{4}f_{2}(x)$ we have

\[
C_{1}=\sqrt{-ip}\dfrac{e^{-\sqrt{-ip}}}{W_{f;p}}\left(f_{2}(1)-f_{2}'(1)\right)\]
\[
C_{2}=-\sqrt{-ip}\dfrac{e^{-\sqrt{-ip}}}{W_{f;p}}\left(f_{1}(1)-f_{1}'(1)\right)\]
\[
C_{3}=-\sqrt{-ip}\dfrac{e^{-\sqrt{-ip}}}{W_{f;p}}\left(f_{2}(-1)+f_{2}'(-1)\right)\]
\[
C_{4}=\sqrt{-ip}\dfrac{e^{-\sqrt{-ip}}}{W_{f;p}}\left(f_{1}(-1)+f_{1}'(-1)\right)\]
Furthermore,  $W_p=[y_+,y_-]=(C_1C_4-C_2C_3)[f_1,f_2]$ is given by
\begin{multline}
  W_{p}=-\frac{e^{-2\sqrt{-ip}}}{W_{f;p}}\bigg(\sigma(f_{1}(1)f_{2}(-1)-f_{1}(-1)f_{2}(1))-f_{1}'(1)f_{2}'(-1)+\\f_{1}'(-1)f_{2}'(1)
  -\sqrt{-ip}(-f_{1}'(-1)f_{2}(1)-f_{1}'(1)f_{2}(-1)+f_{1}(1)f_{2}'(-1)+f_{1}(-1)f_{2}'(1)\bigg)\label{eq:wp}
\end{multline}\label{lemma1}}
\end{Note}
\begin{Proposition}[WKB solutions]\label{fff} In $S_+=\{p:\Re(\sqrt{-ip})\ge 0\}$
there exist two linearly independent solutions of (\ref{eq:hom}) of
the form
 \begin{equation}
    f_{1}(x)=e^{-\sqrt{-ip}x}\left(1-\frac{1}{2\sqrt{p}}\int_{0}^{x}\sqrt{i}V(s)ds+\frac{1}{p}g_{1}(x)\right)\label{eq:f1}\end{equation}
\begin{equation}
  f_{2}(x)=e^{\sqrt{-ip}x}\left(1+\frac{1}{2\sqrt{p}}\int_{0}^{x}\sqrt{i}V(s)ds+\frac{1}{p}g_{2}(x)\right)\label{eq:f2}\end{equation}
where $g_1(x),g'_1(x),g_2(x),g'_2(x)$  are  bounded in $p$ as
$p\rightarrow\infty$ in $S_+$. A similar statement
holds  $S_-=\{p:\Re(\sqrt{-ip})\le 0\}$\end{Proposition}
\begin{Note}
  This is in a sense standard WKB; however, since details about the
  regularity of the terms expansion are needed we provide
  a complete proof.
\end{Note}
\begin{proof} We will only prove the conclusion for $g_{1}$, since the
  proof for $g_{2}$ follows analogously.
  Substituting (\ref{eq:f1}) into (\ref{eq:hom}), we obtain the
  equation for $g_{1}$:
  \begin{multline}
    g_{1}''(x)-2\sqrt{-ip}g_{1}'(x)-V(x)g_{1}(x)+i\frac{\sqrt{-ip}}{2}\left(\int_{0}^{x}V(s)ds-V'(x)\right)=0
  \end{multline}
We rewrite this equation as an integral equation for $g_{1}'$:
\begin{multline}
  g_{1}'(x)=e^{2\sqrt{-ip}x}\\\int_{x_{0}}^{x}e^{-2\sqrt{-ip}s}\left[V(s)\int_{0}^{s}g_{1}'(u)du-\frac{\sqrt{i}}{2}\sqrt{p}\left(\int_{0}^{s}V(u)du-V'(s)\right)\right]ds
\end{multline} where $x_{0}=1$ if $-\pi/2<\arg p<3\pi/2$, and $x_{0}=-1$ if $-5\pi/2<\arg
p<-\pi/2$. Note that $|e^{-2\sqrt{-ip}(s-x)}|\leqslant1$ for all $s$
between $x_{0}$ and $x$. Using integration by parts we obtain
\begin{multline}
g_{1}'(x)=\ds-\frac{1}{2\sqrt{-ip}}\left(V(x)\int_{0}^{x}g_{1}'(u)du-e^{-2\sqrt{-ip}(x_0-x)}V(x_0)\int_{0}^{x_0}g_{1}'(u)du\right)\\
+\ds\frac{1}{2\sqrt{-ip}}\int_{x_{0}}^{x}e^{-2\sqrt{-ip}(s-x)}\left(V'(s)\int_{0}^{s}g_{1}'(u)du+V(s)g_{1}'(s)\right)ds\\
-\ds\frac{1}{4i}\left(\int_{0}^{x}V(u)du-V'(x)-e^{-2\sqrt{-ip}(x_0-x)}\left(\int_{0}^{x_0}V(u)du-V'(x_0)\right)\right)\\
+\ds\frac{1}{4i}\int_{x_{0}}^{x}e^{-2\sqrt{-ip}(s-x)}\left(
V(s)-V''(s)\right)ds
\end{multline}

For large $p$, under the norm $||f||=\sup_{x\in[-1,1]}|f(x)|$ the
above integral equation is easily seen to be contractive inside the
ball \[ ||f||\leqslant\sup_{-1\leqslant
  x\leqslant1}\left(\bigg|V(x)\bigg|+\bigg|V'(x)\bigg|+\bigg|V''(x)\bigg|\right)\]

Therefore $g_{1}'(x)$ and $g_{1}(x)=\int_{0}^{x}g_{1}'(u)du$ are both
bounded in $p$ as $p\rightarrow\infty$.
\end{proof}

\begin{Remark}
  Higher order terms in the asymptotic expansion of $f_{1},f_2$ can be similarly obtained, provided that $V$ is sufficiently smooth.
\end{Remark}

Recalling (\ref{eq:psi}), we see that for large $p$, the poles of
$\psi$ can only come from the zeros of $W_{p}$.
Substituting (\ref{eq:f1}) and (\ref{eq:f2}) into
  (\ref{eq:wp}), we see that
\begin{equation}
  W_{p}=\frac{1}{p^{2}h_3(p)}
\left(e^{-4\sqrt{-ip}}h_{1}(p)+h_{2}(p)\right)\label{eq:wp1}\end{equation}
where\begin{equation}
  h_{1}(p)=\left(\sqrt{p}\frac{\sqrt{i}}{2}V(1)-g_{1}'(1)\right)\left(\sqrt{p}\frac{\sqrt{i}}{2}V(-1)+g_{2}'(1)\right)\label{eq:h1}\end{equation}
 \begin{equation}
  h_{2}(p)=4ip^{3}+2i\sqrt{i}\left(V(1)-V(-1)\right)p^{5/2}+O(p^{2})\label{eq:h2}\end{equation}

and \begin{equation} h_3(p)=-2\sqrt{-ip}+o(1) \label{eq:h33}\end{equation}

\begin{Proposition} \label{asy}In the generic case
when $h_1\not\equiv 0$,
$W_{p}$ has infinitely many zeros in the left half
  plane. Their asymptotic behavior is
  \begin{equation}
    \label{posp}
    p=\left\{
               \begin{array}{ll}
                 -\frac{\pi^{2}i}{4}k^2-\pi k\log k+a_vk+o(k), & \hbox{$V(1)V(-1)\neq 0$;} \\
                 -\frac{\pi^{2}i}{4}k^2-\frac{5\pi}{4} k\log k+b_vk+o(k), & \hbox{exactly one of $V(\pm1)$ is zero;} \\
                 -\frac{\pi^{2}i}{4}k^2-\frac{3\pi}{2} k\log k+c_vk+o(k), & \hbox{$V(1)=V(-1)=0$.}
               \end{array}
             \right.
  \end{equation} where $k\in\mathbb{N}$ and $k\to\infty$, and $a_v,b_v,c_v$ are constants.
\end{Proposition}
\begin{proof}
The equation $W_p=0$ reads
\begin{equation}
  \label{eq:zeros}
  e^{-4\sqrt{-ip}}=-\frac{h_{2}(p)}{h_{1}(p)}
\end{equation}
A simple analysis shows that this can only happen if $p$ is near the
negative imaginary line with $p\sim-k^{2}\pi^{2}i/4$ where
$k\in\mathbb{N}$. We let $p=-i(k\pi/2+z)^{2}$ and rewrite (\ref{eq:zeros}) in terms of $z$:
\begin{equation}
  z=\frac{1}{4i}\log\left(-\frac{h_{2}(-i(k\pi/2+z)^{2})}{h_{1}(-i(k\pi/2+z)^{2})}\right)\label{eq:con}\end{equation}

Recalling (\ref{eq:h1}) and (\ref{eq:h2}), we easily see that the
right hand side of the above equation is contractive for large $k$.

One can find the asymptotic behavior of $z$ by iteration. First
assume $V(1)V(-1)\neq 0$. It is easy to see that
$$-\frac{h_{2}(-i(k\pi/2)^{2})}{h_{1}(-i(k\pi/2)^{2})}=\frac{\pi^4
k^4}{4V(1)V(-1)}(1+O(1/k))$$ Therefore $z\sim-i\log k$. Further
iteration implies $z=-i\log k+\tilde{a}_v+o(1)$.

Similarly, if exactly one of $V(\pm1)$ is zero, then
$z=-\frac{5}{4}i\log k+\tilde{b}_v+o(1)$. If $V(1)=V(-1)=0$ then
$z=-\frac{3}{2}i\log k+\tilde{c}_v+o(1)$; Eq. (\ref{posp}) follows.

\end{proof}

The above analysis shows that all zeros of $W_p$ for large $p$ are
in the left half plane. Thus we have
\begin{Corollary}
There are only finitely many bound states (this, of course
can be simple shown by standard spectral techniques).
\end{Corollary}

We may now proceed to consider the order
of these poles as well as their residues.

\begin{Proposition}\label{P8} The poles of $\hat{\psi}$ for large $p$ are
  simple, and the residues grow sub-exponentially. The residues of $1/W_p$ grow at most polynomially. (In fact, they
  grow exactly polynomially, since the asymptotic expansions in
  (\ref{eq:wp1})--(\ref{eq:h33}) are
  differentiable.) \end{Proposition}

\begin{proof}  Recalling (\ref{eq:wp1}), we notice that
 \[
W_{p}=W'_p(p_k)(p-p_k)(1+o(1))\] where it can be easily checked that
$W'_p(p_k)\ne 0$. Then  $1/W_p'(p_k)$ grows at most polynomially,
and this together with the bounds on $ y_{\pm}$ (see Lemma
\ref{y12}) show that the residues of $\hat{\psi}$ are bounded by
$O(e^{|\Re\sqrt{-ip}|(|x|+2)})$.
\end{proof}

The polynomial growth of residues, along with the analyticity of
$\hat{\psi}$, show convergence of the sum in (\ref{eq:dec}) as
well as its Borel-summability.

\subsection{Asymptotics of $\hat{\psi}$} \label{ilt1}
We will show that $\hat{\psi}$ has sufficient decay to allow for
inverse Laplace transform as well as the desired bending of contour
leading to Borel summation. First we rewrite (\ref{eq:ps1}) as
\begin{equation}
  -iW_{p}\hat{\psi}(x,p)=y_{-}(x)\int_{M}^{x}y_{+}(s)\psi_{0}(s)ds-y_{+}(x)\int_{-M}^{x}y_{-}(s)\psi_{0}(s)ds\label{eq:pso}
\end{equation}
assuming supp\,$\psi_0\in[-M,M]$.

\begin{Lemma}\label{y12}
$y_{\pm}=O\left(\sqrt{p}e^{2|\Re{\sqrt{-ip}}|(|x|+2)}\right)$ for
large $p\in\mathbb{C}$.
\end{Lemma}



\begin{proof}
We will prove the lemma for $y_+$ using matching conditions. The
proof for $y_-$ follows analogously.

The result is obviously true for $x>1$, where
$y_{+}(x)=e^{-\sqrt{-ip}x}$.

For $-1\leq x\leq 1$, we have
$$y_{+}(x)=C_{1}f_{1}(x)+C_{2}f_{2}(x)$$ where

\[
C_{1}=\sqrt{-ip}\dfrac{e^{-\sqrt{-ip}}}{W_{f;p}}\left(f_{2}(1)-f_{2}'(1)\right)\]
\[
C_{2}=-\sqrt{-ip}\dfrac{e^{-\sqrt{-ip}}}{W_{f;p}}\left(f_{1}(1)-f_{1}'(1)\right)\]

and $W_{f;p}=f_{1}(x)f_{2}'(x)-f_{1}'(x)f_{2}(x)$.

It is easy to see, using Proposition \ref{fff}, that
$C_{1,2}=O(\sqrt{p})$. The bounds follow from (\ref{eq:f1}) and
(\ref{eq:f2}).

For $x<-1$ we have
$$y_+(x)=C_5e^{\sqrt{-ip}x}+C_6e^{-\sqrt{-ip}x}$$
where
$$C_5=\frac{1}{2\sqrt{p}}e^{\sqrt{-ip}}(\sqrt{p}y_+(-1)+\sqrt{i}y_+'(-1))$$
$$C_6=\frac{1}{2\sqrt{p}}e^{-\sqrt{-ip}}(\sqrt{p}y_+(-1)-\sqrt{i}y_+'(-1))$$

The result is shown by estimating $y_+(-1)$ and $y_+'(-1)$ with
Proposition \ref{fff}.
\end{proof}

\begin{Lemma}
$p^{-1}W_{p}\hat{\psi}(x,p)=O\(e^{2|\Re{\sqrt{-ip}}|(|x|+M+2)}\)$
 for large
$p\in\mathbb{C}$.

\end{Lemma}
\begin{proof}
A straightforward estimate from (\ref{eq:pso}) and Lemma \ref{y12}.
\end{proof}






\begin{Lemma} (i) There exists a set of curves $p_k(s)$ parameterized by $s\in[0,1]$
with $p(0)$ on the negative imaginary axis, $p(1)$ on the negative
real axis, $|p_k(s)|\geq k$, so that $1/W_p$ is bounded uniformly in
$k$ by a polynomial in $p$ along these curves. Here $k\in\NN$ can be
chosen to be arbitrarily large.

(ii)Moreover, $1/W_p$ is bounded by a polynomial in the region $\{p:
\arg p\in[-\pi,\alpha]$  and $|p|>P_{\alpha}\}$ where
$-\pi<\alpha<-\pi/2$ and $P_{\alpha}>0$ depends only on $\alpha$.

\end{Lemma}

\begin{proof}
We rewrite (\ref{eq:wp1}) as
$$W_{p}=\frac{h_2(p)}{p^{2}h_3(p)}\left(e^{-4\sqrt{-ip}}\frac{h_{1}(p)}{h_{2}(p)}+1\right)$$

We only need to show that
$$\left|e^{-4\sqrt{-ip}}\frac{h_{1}(p)}{h_{2}(p)}+1\right|\geq1$$ on a
chosen set of curves.

Recalling the asymptotic expressions for $h_{1,2}$, we have
$\frac{h_{1}(k^2\pi^{2}i/4)}{h_{2}(k^2\pi^{2}i/4)}\sim c_0k^n$ where
$c_0, n$ are constants.

Let $p_k(s)=-i(k\pi/2-\text{$\frac{1}{4}$}\arg c_0)^{2}(1-is)^2$ and we have
$-4\sqrt{-ip_k}=(2k\pi i-i\arg c_0)(1-is)=(2k\pi -\arg c_0)s+2k\pi
i-i\arg c_0$.

Thus for $s\in[0,1/\sqrt{k}]$ we have
$$e^{-4\sqrt{-ip}}\frac{h_{1}(p)}{h_{2}(p)}\sim e^{2k\pi
s}|c_0|k^n$$ for all $k\in\mathbb{N}$, while for $s\in[1/\sqrt{k},1]$ we have
$$\left|e^{-4\sqrt{-ip}}\frac{h_{1}(p)}{h_{2}(p)}\right|\geqslant e^{2\sqrt{k}\pi}$$
for all $k>0$.

The second part of the lemma follows from the above inequality since
$k$ may be taken to any large real number. Note also that
$\Re(-\sqrt{-ip})=(k\pi/2 -\text{$\frac{1}{4}$}\arg c_0)s$.

It is easy to see that $p_k(s)$ also satisfy the other conditions
specified in the lemma.

\end{proof}

\begin{figure}[ht!] \label{fig:3}
\includegraphics[scale=0.5]{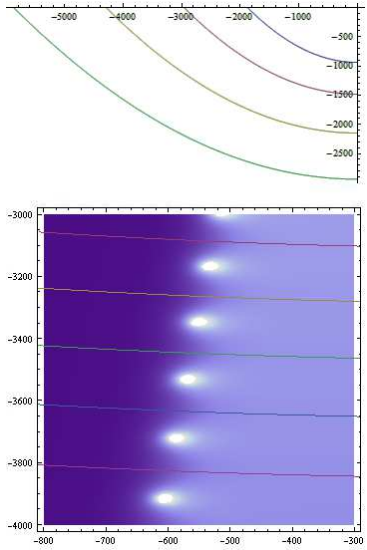}
\caption{Curves $p_k(s)$ passing between poles (plotted with square
barrier potential)}
\end{figure}

Collecting the above results we obtain
\begin{Lemma}
$\hat{\psi}(x,p)=O\(e^{2|\Re\sqrt{-ip}|(|x|+M+2)}\)$ for large p, in
any given sector $\arg p\in[-\pi,\alpha],|p|>P_{\alpha}$ where
$-\pi<\alpha<-\pi/2$ as well as along curves $p_k(s)$ as shown in
the previous lemma.
\end{Lemma}

\subsection{The inverse Laplace transform} \label{ilt}
To obtain the transseries \hyphenation{trans-series} of $\psi$ from
our  $\hat{\psi}$, we take the inverse Laplace transform,
$\frac{1}{2\pi
i}\int_{a_0-i\infty}^{a_0+i\infty}e^{pt}\hat\psi(p)dp$ and push the
contour into the left half plane. We will justify this procedure in
this section.

First we rewrite (\ref{eq:lap}) as an integral equation
$$y=\mathcal{T}(Vy+i\psi_0)$$
where
\begin{multline}\label{ttt}
  \mathcal{T}(f)(x):=\frac{1}{2\sqrt{-ip}}e^{\sqrt{-ip}x}\int_{\infty}^x
  e^{-\sqrt{-ip}s}f(s)ds\\
  -\frac{1}{2\sqrt{-ip}}e^{-\sqrt{-ip}x}\int_{-\infty}^x
  e^{\sqrt{-ip}s}f(s)ds
\end{multline}

We further let $y(x)=\mathcal{T}(i\psi_0)(x)+p^{-3/2}h(x)$ and
rewrite the integral equation as
\begin{equation}\label{eq:hhh}
  h=p^{3/2}\mathcal{T}(V\cdot\mathcal{T}(i\psi_0))+\mathcal{T}(Vh)
\end{equation}
We start with a simple observation.
\begin{Remark}
$e^{-\sqrt{-ip}}$ is bounded in the region
$\Omega:=\{p\in\mathbb{C}:-\pi/2\leq\arg p\leq
\pi\}\bigcup\{p\in\mathbb{C}:-\Im p>{\rm const}\,\,(\Re p)^2$. (In
the following, we will choose ${\rm const}=1/9$)
\end{Remark}
 Note that
$$\Re (-\sqrt{-ip})=-\frac{1}{2\Im p}\sqrt{\Re p+|p|}(\Im p-\Re
p+|p|)$$
We denote $\mu=\sup_{p\in\Omega}|e^{-\sqrt{-ip}}|$.

\begin{Lemma}\label{fg}
Assume $f$ and $g$ are locally bounded functions and $fg$ is
compactly supported, with supp$\,(fg)\in[-b,b]$ where $b>0$. Let
$a\geq b$ be an arbitrary number, $\Omega'=\Omega\bigcup
\{p\in\mathbb{C}:|p|>p_v>1\}$. We then have
$$|\mathcal{T}(fg)|\leqslant\frac{2b\mu^b \sup_{x\in[-b,b]}
|g(x)|}{\sqrt{p_v}}||f||$$ where
$||f||:=\sup_{p\in\Omega',x\in[-a,a]}|f(x,p)|$

\end{Lemma}

\begin{proof}
By (\ref{ttt}) we have
\begin{multline}
  |\mathcal{T}(fg)(x,p)|\leqslant \frac{1}{2|\sqrt{p}|}\int_{0}^b
  |e^{-\sqrt{-ip}u}||g(u+x)||f(u+x)|du\\
  +\frac{1}{2|\sqrt{p}|}\int_{-b}^0
  |e^{\sqrt{-ip}u}||g(u+x)||f(u+x)|du\\
  \leqslant \frac{\mu^b}{\sqrt{p_v}}\int_{-b}^{b}|g(s)||f(s)|ds \leqslant
  \frac{2b\mu^b \sup_{x\in[-b,b]}
|g(x)|}{\sqrt{p_v}}||f||
\end{multline}
\end{proof}
\begin{Lemma}\label{p00}
For compactly supported and twice differentiable $\psi_0$, we have
$$\mathcal{T}(i\psi_0)(x)=\frac{1}{p}\psi_0(x)+\frac{1}{p^{3/2}}G_1(x,p)$$
where $|G_1(x,p)|\leq 2M\sup |\psi_0''|
\sup_{s\in[0,M+|x|]}|e^{-\sqrt{-ip}s}|$, assuming
supp\,\,$\psi_0\in(-M,M)$.
\end{Lemma}

\begin{proof}
This is shown by repeated integration by parts to (\ref{ttt}). Note
that
\begin{multline}
\mathcal{T}(i\psi_0)(x)=\frac{1}{p}\psi_0(x)-\frac{1}{2p}e^{\sqrt{-ip}x}\int_{M}^x
  e^{-\sqrt{-ip}s}\psi_0'(s)ds\\-\frac{1}{2p}e^{-\sqrt{-ip}x}\int_{-M}^x
  e^{\sqrt{-ip}s}\psi_0'(s)ds\\=\frac{1}{p}\psi_0(x)+\frac{1}{2(ip)^{3/2}}e^{\sqrt{-ip}x}\int_{M}^x
  e^{-\sqrt{-ip}s}\psi_0''(s)ds\\-\frac{1}{2(ip)^{3/2}}e^{-\sqrt{-ip}x}\int_{-M}^x
  e^{\sqrt{-ip}s}\psi_0''(s)ds\\
  \end{multline}
  \end{proof}
\z  With
the above lemmas, we have
\begin{Proposition}
Let $\Omega_0=\Omega\bigcup \{p\in\mathbb{C}:|p|>p_v\}$ where
$p_v=9(\sup_{x\in[-1,1]} V(x)+\mu+1)^2$. Let $x_1>0$ be an arbitrary
real number. The integral equation (\ref{eq:hhh}) is contractive in
the space of functions analytic in $p\in\Omega_0$ equipped with the
sup norm $||f||=\sup_{p\in\Omega_0,x\in[-x_1,x_1]}|f(x,p)|$, within
a ball of size $$2\mu \sup_{x\in[-1,1]} |V(x)|+2M\sup |\psi_0''|
\sup_{s\in[0,M+|x_1|]}|e^{-\sqrt{-ip}s}|$$ In particular, the
solution $h$ is bounded as $x_1\to\infty$ if $\Re p>0$. See Fig. \ref{fig:2}.

\end{Proposition}

\begin{proof}
The estimates of $p^{3/2}\mathcal{T}(V\cdot\mathcal{T}(i\psi_0))$
follow from lemma \ref{p00} and \ref{fg}, with $f=V$, $g=\psi_0$
and $g=G_1$ separately.

The contractivity of $\mathcal{T}$ follows from lemma \ref{fg}
with $f=V$, $g=h$. Note that analyticity in $p$ is preserved by
$\mathcal{T}$ and convergence in the sup norm.
\end{proof}

\z We therefore have the following results.
\begin{Proposition}\label{pro}
  (i) $\hat{\psi}$ (as in section \ref{main}) has the following
  decomposition:
$$\hat{\psi}(x,p)=\frac{1}{p}\psi_0(x)+\frac{1}{p^{3/2}}G_2(x,p)$$ where $G_2(x,p)$ is bounded in
$p\in\Omega_0,x\in[-x_1,x_1]$ where $x_1>0$ is arbitrary.\\
(ii) $$\psi(x,t)=\psi_0(x)+\frac{1}{2\pi
  i}\int_{a_0-i\infty}^{a_0+i\infty}\frac{G_2(x,p)}{p^{3/2}}e^{pt}dp$$
is the solution to (\ref{eq:ori}). Here $a_0>0$ is a constant.
\end{Proposition}
\begin{proof}
  We only need to show that in $\Omega_0$ the solution $\hat{\psi}$ is identical
  to the solution $y$ obtained in this section, the decomposition of
  which has already been shown. Part (ii) then follows immediately from
  properties of the inverse Laplace transform.

  To this end, note that the general solution to (\ref{eq:lap}) can be
  written in the form of
$$y_{gen}(x,p)=\hat{\psi}(x,p)+c_1(p)y_+(x,p)+c_2(p)y_-(x,p)$$ where
$y_+$ and $y_-$ are the homogeneous solutions defined in section
\ref{main} (with a slight abuse of notation). This implies
$$y(x,p)=\hat{\psi}(x,p)+c_1(p)y_+(x,p)+c_2(p)y_-(x,p)$$ where
$y(x,p)$ is the solution obtained earlier in this section. 
Since in the region $\{\Re(p)>p_v$, $x<1\}$, $y_+$ is unbounded and
$y_-$ is bounded, and in $\{\Re(p)>p_v$, $x>1\}$ $y_+$ is bounded
and $y_-$ is unbounded, while both $y$ and $\hat{\psi}$ are bounded
(the boundedness of $\hat{\psi}$ follows easily from
(\ref{eq:pso})), we must have $c_1=c_2=0$ in $\Re(p)>p_v$. Thus
$\hat{\psi}$ and $y$ coincide in $\Re(p)>p_v$ and also in $\Omega_0$
by uniqueness of analytic continuation.

\end{proof}

\begin{Lemma}
In the expression $$\frac{1}{2\pi
  i}\int_{a_0-i\infty}^{a_0+i\infty}\frac{h_2(x,p)}{p^{3/2}}e^{pt}dp$$
we may deform the contour to one which goes from $-\infty$ below the
real axis, turns counterclockwise around the origin and goes towards
$-\infty$ above the real axis. In the process we collect all
residues from all the poles in the left half plane.
\end{Lemma}

\begin{proof}
The deformation of the upper half of the contour is justified by
Proposition \ref{pro} since $\Omega$ contains the second quadrant.

In the third quadrant, recall that
$p_k(s)=-i(k\pi/2-\text{$\frac{1}{4}$}\arg c_0)^{2}(1-is)^2$ and
$\Re(\sqrt{-ip_k})=(k\pi/2 -\text{$\frac{1}{4}$}\arg c_0)s$. Thus
$\Re p_k(s)\sim const. k^2 s$ and $\Re(\sqrt{-ip_k})=O(\Re
p_k(s)/k)$ for all $s>0$. Therefore we may choose part of the curve
$p_k(s)$ where $s\in[1/k,1]$ and join it with a curve in
$\{p\in\mathbb{C}:-\Im p>(\Re p)^2/9\}\bigcap\Omega_0$, say a
vertical line downward to infinity. These two curves, along with the
one from below the real axis to the origin and lower half of the
original contour, surround all poles in the third quadrant as
$k\to\infty$. Decay along the $p_k(s)$ curve is ensured by the term $e^{pt}$, since $e^{p_k(s)t\pm
\sqrt{-ip_k(s)}M_0}=O(e^{-kt})$ for arbitrarily large $M_0$. Note
also that the length of $p_k(s)$ is of order $k^2$.
\end{proof}

\begin{figure}[ht!] \label{fig:2}
\includegraphics[scale=0.5]{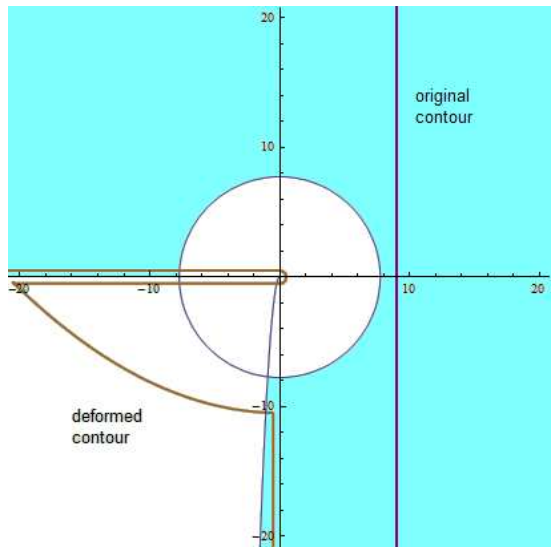}
\caption{A sketch of the contour deformation used. The shaded
area is the contractivity
region.}
\end{figure}

To prove Theorem \ref{T1}, we further write $\psi_0(x)=\frac{1}{2\pi
  i}\int_C \frac{1}{p}\psi_0(x)$ and combine it with $\frac{1}{2\pi
  i}\int_C\frac{h_2(x,p)}{p^{3/2}}e^{pt}dp$, where the contour of
  integration is the horizontal part around the negative real axis described
  above. This contour can be deformed to "$0$ to $-\infty$" in an
  upper and a lower sheets of the Riemann surface, which yields a Borel-summable power series in $t^{-1/2}$.

\subsection{Connection with Gamow Vectors}
Classically, Gamow vectors are obtained as solutions to
(\ref{eq:hom}) with ``purely outgoing boundary conditions'' as $x\to
\pm\infty$. In our case, this means such a solution (after
rescaling) equals $y_+(x)$ for $x>1$ and a nonzero constant times
$y_-(x)$ for $x<1$, $y_{\pm}$ being as in section (\ref{main}). The
existence of such a solution, therefore, is equivalent to the linear
dependence of $y_+$ and $y_-$ (cf. Lemma \ref{lemma1}), which in
turn is equivalent to the vanishing of the Wronskian: $W_p=0$.
Thus the $\gamma_k$
 found from the poles of $\hat{\psi}$ are exactly the resonances
corresponding to the Gamow vectors, a constant multiple of $y_+$.
The latter are is easily seen to be multiples of the residues of
$\hat{\psi}$ for example by simplifying (\ref{eq:pso}):

$$-iW_{p}\hat{\psi}(x,p)=cy_+(x)\int_{M}^{x}y_{+}(s)\psi_{0}(s)ds-y_{+}(x)\int_{-M}^{x}cy_{+}(s)\psi_{0}(s)ds$$
$$=-c\left(\int_{-M}^{M}y_{+}(s)\psi_{0}(s)ds\right)y_+(x)$$
\subsection{Proof of Proposition ~\ref{approx2}}
\begin{proof}
  This follows straightforwardly from Lemmas \ref{L18} and \ref{L19},
after extracting a suitable number of poles from $\hat{\psi}$. All poles
are simple, and the contribution of a pole of residue $r_k$ and
position $p_k$ is
$$J_k(t)=r_k\int_0^{\infty}\frac{e^{-pt}dp}{\sqrt{p}-p_k}$$
The representation of $J_k$ in terms of special functions is perhaps most conveniently
shown by solving the first order differential equation it satisfies,
and determining the free constant from the asymptotic
behavior in $p_k$.
\end{proof}
\section{Example: the square barrier}\label{SB}
Here we take as a simple example the Schr\"{o}dinger
equation with a square bump potential $V(x)=\chi_{[-1,1]}$, $\chi$
being the indicator function. (One of few cases where explicit solutions exist.)
$$y_+(x)=\left\{
           \begin{array}{ll}
             A_1 e^{\sqrt{-ip}x}+A_2 e^{-\sqrt{-ip}x}, & \hbox{$x\leqslant -1$;} \\
             A_3 e^{\sqrt{1-ip}x}+A_4 e^{-\sqrt{1-ip}x}, & \hbox{$-1<x<1$;} \\
             e^{-\sqrt{-ip}x}, & \hbox{$x\geqslant 1$.}
           \end{array}
         \right.
$$

$$y_-(x)=\left\{
           \begin{array}{ll}
            e^{\sqrt{-ip}x}, & \hbox{$x\leqslant-1$;} \\
             B_1 e^{\sqrt{1-ip}x}+B_2 e^{-\sqrt{1-ip}x}, & \hbox{$-1<x<1$;} \\
             B_3e^{\sqrt{-ip}x}+B_4e^{-\sqrt{-ip}x}, & \hbox{$x\geqslant1$.}
           \end{array}
         \right.$$
where the coefficients $A_j$, $B_j$ are determined by matching
solutions at the endpoints, $\pm1$. For example,
$$A_3=\frac{\sqrt{i+p}-\sqrt{p}}{2\sqrt{i+p}}(e^{-\sqrt{-ip}-\sqrt{1-ip}})$$
The other coefficients are similar (and obtained
in a similar way) and we omit them.

It follows that the Wronskian $W_p$ has an explicit expression
\begin{equation}
\frac{\sqrt{-i}e^{-2\sqrt{-ip}+2\sqrt{1-ip}}}{2\sqrt{i+p}}\left(e^{-4\sqrt{1-ip}}(i+2p-2\sqrt{p}\sqrt{i+p})-i-2p-2\sqrt{p}\sqrt{i+p}\right)\label{wro}\end{equation}

We may find the asymptotic positions of the resonances by iterating
\begin{equation}
  z_k=\frac{1}{4i}\log\left(\frac{i+2p_k+2\sqrt{p_k}\sqrt{i+p_k}}{i+2p_k-2\sqrt{p_k}\sqrt{i+p_k}}\right)\end{equation}
where $p_k=-i(k\pi/2+z_k)^{2}$.

We also calculate the residues of $1/W_p$ by differentiating
(\ref{wro}):
$$1/W_p\sim\frac{\sqrt{p_k}(i+p_k)(i+2p_k-2\sqrt{p_k}\sqrt{i+p_k})}{\sqrt{-i}e^{-2\sqrt{-ip_k}+2\sqrt{1-ip_k}}(1+\sqrt{-ip_k})}\frac{1}{p-p_k}$$

Here we calculate the positions and residues of a series of poles
using the above formulas and compare them  to the asymptotic
behavior $-\pi k \log(\pi k)-i\pi^2k^2/4$, as in
Proposition \ref{asy}. Then we plot these poles together with a
density graph.\\

The asymptotic pole location formula gives (increasingly) good
accuracy starting with the $15$th pole or so, where it predicts the
position $-181-555i$, whereas the exact value is about $-180-532i$.

\begin{figure}[ht!] \label{fig:4}
\includegraphics[scale=0.6]{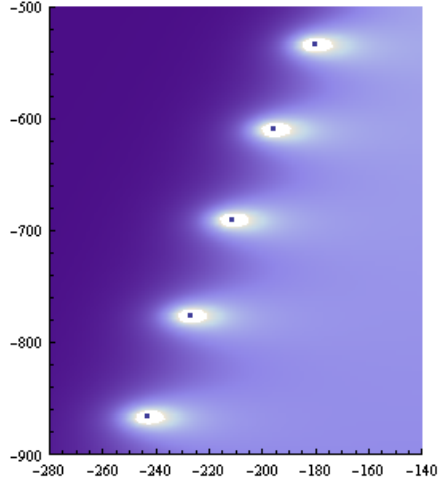}
\caption{Density graph of $1/W_p$. Dark dots indicate poles
calculated from the asymptotic formula.}
\end{figure}

The first resonance, the one closest to the imaginary line (in
$p$-plane), may have a visible effect on the wave function $\psi$
even if this resonance does not correspond to a (long-lived)
metastable state. We will demonstrate this phenomenon, as well as
the computational effectiveness of the Borel summation approach,
using (near-) optimal truncation,  see \S\ref{numcalc}, on the
example of the square barrier potential, where we choose the initial
condition to be $\psi_0(x)=\chi[-\frac{1}{2},\frac{1}{2}]$ for
simplicity.

In our example, the first pole of $1/W_p$ is located at
$p_0=-1.70018 - 0.805871i$. This can be found by standard iterative
arguments.

We will demonstrate the effect of this pole in the region $x>1$,
where (cf. (\ref{eq:ps1}))
$$\hat{\psi}(x,p)=-\frac{ie^{-\sqrt{-ip}x}}{W_p}\int_{-\frac{1}{2}}^{\frac{1}{2}}y_-(s;p)ds$$
and
\begin{multline}\label{fin}
\psi(x,t)=-\frac{1}{2\pi}\int_{-i\infty}^{i\infty}\frac{e^{pt-\sqrt{-ip}x}}{W_p}\int_{-\frac{1}{2}}^{\frac{1}{2}}y_-(s;p)dsdp\\
=-\frac{1}{2\pi}\int_{-\infty}^{0}\frac{e^{pt-\sqrt{-ip}x}}{W_p}\int_{-\frac{1}{2}}^{\frac{1}{2}}y_-(s;p)dsdp-\frac{1}{2\pi}\int_{0}^{-\infty}\frac{e^{pt-\sqrt{-ip}x}}{W_p}\int_{-\frac{1}{2}}^{\frac{1}{2}}y_-(s;p)dsdp\\
+\frac{1}{2\pi i}\lim_{p\to
p_0}\((p-p_0)\hat{\psi}(x,p)\)e^{p_0t}\(1+o(1)\)
\end{multline}

for large $t$.

We may calculate the power series by expanding
$\hat{\psi}(x,p)$ near $p=0$ and using Watson's Lemma. For instance,
for $x=8$ we obtain the series
\begin{multline}(0.735266+ 0.735266 i) \frac{1}{t^{
  3/2}} - (12.3883- 12.3883 i) \frac{1}{t^{
  5/2}} - (98.5277+ 98.5277 i) \frac{1}{t^{
  7/2}}\\ + (471.935- 471.935 i) \frac{1}{t^{
  9/2}} + (1429.08+ 1429.08 i) \frac{1}{t^{
  11/2}}\\ - (2690.72- 2690.72 i) \frac{1}{t^{
  13/2}} - (4000.95+ 4000.95 i) \frac{1}{t^{
  15/2}}+O(t^{-8})\end{multline}

\begin{Note}
 Taking $x=7$, the contribution of the first resonance to the power series
is visible,  about $3\%$,  for $t=7$, and evidently decreases exponentially thereafter. The overall precision
increases rapidly with $t$, if $x$ is fixed or does not increase faster than $t$,  unlike most direct numerical calculations.
\end{Note}

\section{Ecalle-Borel summation, exponential
  asymptotics}\label{transth}
An expansion is Borel summable if, by definition, it is the
asymptotic expansion of the Laplace transform $\mathcal{L}$ of a
function which is real-analytic on $\RR^+$, exponentially bounded,
and which has a convergent series at zero, in (ramified) powers of
the variable and possibly logs (Frobenius series). The Borel
summation operator, $\mathcal{LB}$, is essentially
$\mathcal{L}S\mathcal{L}^{-1}$ where $\mathcal{L}^{-1}$ applied to a
series is understood in the formal sense, as the term-by-term
transform, and $S$ is convergent summation. Since $\mathcal{LB}$ is
conjugated with usual summation, which commutes with virtually all
operations, the same is true for $\mathcal{LB}$. Uniqueness of the
Borel sum  stems from uniqueness of the sum of a
convergent series. Borel summation is a canonical extension of usual
summation. The decomposition of a function in a Borel summed part
and a sum of exponentials, when possible, is also unique and
canonical, see \S\ref{transth} and for a detailed analysis e.g.
\cite{CostinBook}.

\subsection{Borel summation and least term truncation}\label{numcalc} Borel summation
allows for exponentially accurate calculations of the associated
function by truncating the series near its least term. We first
briefly explain the reason and refer to \cite{OCMDK} for more details.
The accuracy, even for $t$ not so large, is illustrated
in \S\ref{SB}.

Here we consider a Borel summed series of the type that intervenes
in our problem, namely let
\begin{equation}
  \label{eq:bor1}
  f(t)=\int_0^{\infty} e^{-pt} F(p)dp
\end{equation}
where $F(p)=g(p^\beta); \Re\beta>0$ and $g$ is analytic at the origin and
meromorphic in $\CC$. The asymptotic behavior of $f$ is, by Watson's
lemma, of the form
\begin{equation}
  \label{eq:srf}
  f(t)\sim \tilde{f}(t)=\sum_{k=0}^{\infty}c_k t^{-k\beta}
\end{equation} We let $S_n(t;\tilde{f})$ be the truncate
of the power series, up to the power $t^n$:
\begin{equation}
  \label{eq:defS}
  S_n(t;\tilde{f})=\sum_{0\le k\Re\,\beta\le n}c_k t^{-k\beta}
\end{equation}
We want to estimate the error by calculating the function from its
power series by optimal truncation, or {\em truncation to the least
  term}. For a series in which the
coefficient of $t^b$ grows roughly like $ b!/c^b$,
means using the  truncation
$S_{ct}(t;\tilde{f})$.  For example using Stirling's formula we see
that for $F(p)=1/(p-p_0)$, the general term of $\tilde{f}$ grows like
$k!/|p_0|^k$ and its least term is near $n=t|p_0|$; this location of
the least term is the same regardless of the nature of the
singularity, for all algebraic-logarithmic type of singularities. 

We show that the error in approximating $f$ by $S_t(t;\tilde f)$ in this way is
of the same order of magnitude as this least term, which is
exponentially small in $t$. In similar contexts this is known quite
generally for Borel summed series, see \cite{OCMDK} and references
therein.  Least term truncation provides a practical way to calculate
functions with very high accuracy even for $t$ of moderate size. In
spite of the generality considered in \cite{OCMDK}, our case is not
covered (because of ramification at zero).  Instead of describing the
adaptation of that proof, for convenience of the reader, we provide a complete argument in our case.

\begin{Lemma}\label{L18}
  Assume $\beta\in (0,1)$ and $H$ is analytic in  disk
  $\overline{\mathbb{D}}_{A_2}$. Let $H_M$ be the maximum
of $H$ on a disk of radius $A\in(A_1,A_2) $. Then
\begin{equation}
  \label{eq:eqg11}
  \int_0^{A_1^{1/\beta}}H(p^{\beta})\exp(-tp)dp-S_{(n+1)\beta}(t,\tilde H(p^\beta))
=
O\left(\frac{H_M\Gamma(n\beta+\beta+1)}{A_1^{n+1}t^{n\beta+\beta+1}} \right)
\end{equation}
\end{Lemma}
\begin{proof}
  Indeed, by using Taylor series with Cauchy integral  remainder we have
  \begin{multline}
    \label{eq:eqH1}
 \int_0^{A_1^{1/\beta}}H(p^{\beta})\exp(-tp)dp\,\,-\,\,S_n(H(p^\beta))\\=
\frac{1}{2\pi i}\int_0^{A_1^{1/\beta}}\exp(-tp)p^{(n+1)\beta} \oint_C\frac{H(s)ds}{s^{n+1}(s-p^\beta)}dp=J
  \end{multline}
where $C$ is a circle of radius $A$.
If $H_M$ is the maximum of $H$ on $C$, We have
\begin{equation}
  \label{eq:eqJ}
  |J|\le \frac{H_M}{2\pi A_1^{n}(A_1-A)}\int_0^{\infty} e^{-np}p^{(n+1)\beta}dp
\end{equation}
and the result follows.
\end{proof}
\begin{Lemma}
  Assume $\beta\in (0,1)$ and let $F(p)=(p^\beta-p_0)^{-1}$  where $p_0\notin\RR^+$ and let  $f=\mathcal{L}F$. Then,
  \begin{equation}
    \label{eq:stP}
    |f(t)-S_{(n+1)\beta}(t,\tilde f)|\le \frac{\Gamma(n\beta+\beta+1)}{p_0^{n+1}t^{n\beta +\beta +1}|\Im\,p_0|}
  \end{equation}
\end{Lemma}\label{L19}
\begin{proof}
 Writing
\begin{equation}
    \label{eq:p0}
    \frac{1}{p^b-p_0}=\sum_{j=0}^n \frac{p^{j\beta}}{p_0^{j+1}}+
\frac{p^{(n+1)\beta}}{p_0^{n+1}(p^\beta-p_0)}
\end{equation}
this follows from straightforward integration and estimates.
\end{proof}
\subsection{A class of level one transseries}\label{Sec1.2} We only need an
especially simple subclass of transseries, exponential power series of
the type
\begin{equation}
  \label{eq:deftr}
  \tilde{f}(t)=\sum_{k=0}^{\infty} e^{-\gamma_k t} t^{\alpha_k}\tilde{f}_k(t)
\end{equation}
where $\tilde{f}_k(t)$ are formal power (integer or noninteger) series
in $1/t$, where, for
disambiguation purposes, the real part of the leading power of $1/t$ in $\tilde{f}_k(t)$
is chosen to be $1$. Agreeing that no $
\tilde{f}_k(t)$ is exactly zero and the $ \gamma_k $ are distinct, it
is required that the exponentials $ e^{-\gamma_k t}$ are well
ordered, in the sense that $\Re(\gamma_k)\ge \Re(\gamma_{k'})$ if
$k\ge k'$, and every $\Re(\gamma_k)$ has a predecessor, the smallest
$\Re(\gamma_j)$ greater than it. In our context the sets
$\{j:\Re(\gamma_j)=\Re(\gamma_k)\}$ turn out to be finite.

The transseries $\tilde{f}$  is Ecalle-Borel summable if
(a) $\tilde{f}_k(t),k\in\NN$ are simultaneously Ecalle-Borel summable (in
fact, simply Borel summable, in our case), and (b) upon replacing each
$\tilde{f}_k(t)$ by its sum, the resulting function series is
uniformly convergent. We give precise definitions in \S\ref{transth}.
Transseries and Ecalle-Borel summability were introduced by Ecalle in
the 1980s and there has been substantial development since. For an
elementary introduction see \cite{CostinBook}.

The transseries is (Ecalle-Borel) summable if for some $T>0$
we have the following.

(i) $\tilde{f}_k(t)$ are simultaneously Borel summable, that is there
exists a $T$ so that $\tilde{f}_k(t)$ are the asymptotic expansions
for large $t$ of Laplace transforms,
\begin{equation}
  \label{eq:lptr}
f_k(t)=  \int_0^{\infty}F_k(p)e^{-pt}dp=:\mathcal{LB}\tilde{f}_k(t)
\end{equation}
where

(ii) $F_k$ are ramified-analytic at zero, and real analytic on $\RR^+$
with the uniform bound $\|F_k(p)\|\le C_k e^{|p|T}$.

(iii) For some $\nu\in\RR$ we have have $|F_k(t)|\le C_k e^{\nu|p|}$.

(iv) The series
\begin{equation}
  \label{eq:convtr}
  \sum_{k=0}^{\infty} |e^{-\gamma_k T}| C_k
\end{equation}
converges for some $T>0$ (and thus for all $t\ge T$). We recall that,
by convention, $F_k(p)=c_k(1+o(1))$ as $p\to 0$, where $c_k\ne 0$.

Therefore, the sum
\begin{equation}
  \label{eq:sumtrans}
  f= \mathcal{LB} \sum_{k=0}^{\infty} t^{\alpha_k}e^{-\gamma_k t} \tilde{f}_k(t):=\sum_{k=0}^{\infty} e^{-\gamma_k t} t^{\alpha_k}\mathcal{LB}\tilde{f}_k(t)=\sum_{k=0}^{\infty} e^{-\gamma_k t} t^{\alpha_k}{f}_k(t)
\end{equation}
converges absolutely for $t>T$.

The operator $\mathcal{LB}$ is a proper extension of the Borel
summation operator. In particular, it allows for non-accumulating
singularities on the axis of summation, in which case analytic
continuation is replaced by Ecalle's universal
averaging. Superexponential growth of $F$ of a controlled type is
allowed, using Ecalle's acceleration operators.

 With these extensions,
Borel summation is an extended isomorphism between series, or more
generally transseries, and a class of functions ({\em analyzable}
functions), commuting essentially with all operations with which
analytic continuation does. In this sense, Ecalle-Borel summable
transseries substitute successfully for convergent expansions; in
particular the Ecalle-Borel sum of a formal solution of a problem
(within certain known classes of problems such as ODEs and PDEs) is an
actual solution of the same problem. It is known that the fundamental
decaying solution of a nonlinear differential equation at a generic
singularity is given, uniquely, by Borel summable transseries
\cite{Duke}.

\subsection{Uniqueness of the transseries representation} In the
same way as the asymptotic power series of a function, when a series
exists, is unique one function can only have one transseries
representation, if at all.
We sketch a proof that a representation of the form
(\ref{eq:sumtrans}) of a given $f$ is unique. We assume
of course that the transseries are in canonical form, as explained
above.  By linearity, it suffices to show that if $f$ given in
(\ref{eq:sumtrans}) is identically zero, then all $\tilde{f_k}$, and thus all
$f_k$ are identically zero.  We assume
by contradiction
that some $\tilde{f}_k$ are nonzero. Since the
$\Re(\gamma_k)$ are well ordered, cf. \S\ref{Sec1.2}, we choose the
largest $\Re(\gamma_k)$ such that $\tilde{f}_k\not\equiv 0$. There are
only finitely many $\gamma_k$ with the same $\Re(\gamma_k)$, cf. again
\S\ref{Sec1.2}. We can assume without loss of generality that
these $\lambda$s have indices $0,...,n$, and assume that we have ordered the terms
in the transseries so that $\Re \gamma_i\le \Re\gamma_{i+1}$ for all $i$.
We write
\begin{equation}
  \label{eq:sumtrans2}
  f= \sum_{k=0}^{n} e^{-\gamma_k t} t^{\alpha_k}{f}_k(t)+\sum_{k=n+1}^{\infty} e^{-\gamma_k t} t^{\alpha_k}{f}_k(t)
\end{equation}
Note that for any $\epsilon>0$ small enough we have
\begin{equation}
  \label{eq:sumtrans3}
 \left|\sum_{k=n+1}^{\infty} e^{-\gamma_k(T+\tau)} t^{\alpha_k}{f}_k(t)\right|\le const  |e^{-\gamma_{n+1}\tau}|=o\left( |e^{-\gamma_{0}(T+\tau)}|\right)
\end{equation}
as $\tau \to\infty$, since $\Re(\gamma_0)<\Re(\gamma_{n+1})$. Dividing
(\ref{eq:sumtrans2}) by $e^{-\Re \gamma_0 t}$ we get
  \begin{equation}
    \label{eq:eq56}
   \sum_{k=0}^{n} e^{-i\alpha_k t}  t^{\alpha_k}{f}_k(t) =o(1), \ \ \ (t\to \infty)
  \end{equation}
  where $\alpha_k=\Im\gamma_k$.  For each $k$ we choose $\beta_k$ to be
 the smallest power of $p$ (in
  absolute value) with nonzero coefficient, $c_k$ in the expansion
  of $F_k$. Of course, if all coefficients in the Puiseux series
of $F_k$ vanish, then $F_k$ vanishes near zero, and thus everywhere
by analyticity. We arrange that there is no $k$ such that $F_k\equiv 0$.
Then, by Watson's lemma, $F_k=c_k\Gamma(\beta_k+1)t^{-\beta_k-1}(1+o(1))$ for large $t$. We choose the largest $\beta_j$, in the sense above, and divide
by $\Gamma(\beta_j+1)t^{-\Re \beta_j-1}$. We get, by Watson's Lemma,
\begin{equation}
  \label{eq:58}
   \sum_{\Re \beta_j=\Re \beta_k;k\le n} c_k e^{-i\alpha_k t}t^{-i\theta_k}=o(1)
\end{equation}
where $\theta_k=\Im\beta_k$.
We now prove a lemma in more generality than needed here, in view of
future generalizations to time dependent potentials.

\begin{Lemma}
  Assume $\sum_{k=0}^{\infty}|c_k|^2<\infty$ and that
  \begin{equation}
    \label{eq:eqas1}
   f(t)=  \sum_{k=0}^{\infty} c_k e^{-i\alpha_k t}t^{-i\theta_k}=o(1)
  \end{equation}
where $\alpha_k,\theta_k\in\RR$, as $t\to \infty$. Then $f(t)\equiv 0$.
\end{Lemma}
\begin{proof}
  We first look at the simpler case where all $\theta_k=0$; as we shall
see, the general case is similar.
We see, by explicit integration and
dominated convergence,  that for large $t_0$ and $t\to\infty$ we get
from (\ref{eq:eqas1}) that
\begin{equation}
  \label{eq:equint}
  \int_{t_0}^t |f(s)|^2ds=\sum_{k=0}^{\infty}|c_k|^2t+O\left[\left(\sum_{k=0}^{\infty}|c_k|^2\right)^2\right]=o(t)
\end{equation}
which is only possible if
\begin{equation}
  \label{eq:sum2}
  \sum_{k=0}^{\infty}|c_k|^2=0
\end{equation}
To generalize to the case $\theta_k\ne 0$, we simply note that (\ref{eq:eqas1}) implies
 \begin{equation}
    \label{eq:eqas1}
   f(e^s)=  \sum_{k=0}^{\infty} c_k e^{-i\alpha_k e^s}e^{-i\theta_k s}=o(1)
  \end{equation}
as $s\to \infty$ and that, still as $s\to \infty$ we have (e.g.
by integration by parts) that, for $\theta\ne 0$, we have
\begin{equation}
  \label{eq:eqf3}
  \int_{s_0}^s e^{-i\alpha e^u}e^{-i\theta u}du=\frac{i}{\theta}e^{-i\alpha e^s}e^{-i\theta s}(1+o(1))
\end{equation}
\end{proof}
{\em Borel summation and usual summation: the underlying isomorphism.}
Furthermore, there is the following important point.  When a Borel
summable transseries of a function exists, functions and their
transseries have the same properties. That is, there exists an
extended isomorphism between transseriable functions and transseries
similar in many ways to the one between germs of analytic functions,
and their local convergent Taylor series regarded as formal algebraic
objects. This latter isomorphism is so flawless that we do not
distinguish notationally a convergent sum as a formal sum, from its
sum as a function.  Borel summation is a proper extension of usual
summation, carrying further these isomorphism features.

 The isomorphism, provided by Ecalle-Borel summability, which recovers
the function from its transseries, justifies the usage of the term
{\em complete asymptotics}. Borel summation is a  canonical way to sum factorially
divergent series, cf. also \cite{CostinBook}.

{\em Independence of method.} Finally, the nontrivial terms in the
transseries of a function  can be exhibited by many other exponential
asymptotic techniques some of which having
of substantial calculational  value, such as
hyperasymptotics, a set of methods \hyphenation{Hyper-asymptotics}
improving and refining optimal truncation of series, cf.
\cite{Howls},\cite{CK}, \cite{Daalhuis}, and references therein.

In the language of generalized Borel summability,
the wave function asymptotics is given in all amplitude
regimes by an Ecalle-Borel summable transseries, valid for $t>0$, and
this transseries turns out to rest on a Gamow vector decomposition.
\begin{Note}\label{N6} {\rm
Sometimes a given series can be Borel summed with respect to different
powers, or more generally functions, of the variable. For instance,
\begin{equation}
  \label{eq:eqnnu}
  \int_0^{\infty}e^{-pt}e^{-ip^2/4}dp=\int_0^{\infty}e^{-pt^2}
\frac{dp}{2\sqrt{\pi p}(i-p)}
\end{equation}
Since the integrals are equal, they have the same asymptotic series
for large $t$; both integrals are Borel sums of the {\em same
  asymptotic series}, on the left interpreted as a series in $1/t$, while
 on the left it is thought of as a series in $1/t^2$.

Using the connection with Gevrey asymptotics,
\cite{CostinBook}, it is easy to see that a series has a unique Borel
sum, with respect to any variable {\em in which it is Borel summable}, even
when allowing for ramified-analytic functions.   Ramified-analytic
 functions  are real analytic,
and near $p=0$
of the form $F(p_1,...,p_m)$ where $p_j=p^{a_j}\log p^{b_j}$, $\Re(a_i)>0$, and $F$ is analytic at $\bf 0$ .  It is an easy exercise to show that
$ e^{-t^\alpha}, \,\,\Re\alpha>0$ {\em cannot} be represented as a Laplace transform of a ramified analytic
function.

But beyond ramified analyticity  uniqueness of the
representation as a ``continuum'' (to use physics terminology)  plus exponentials does not hold. We have, e.g.,
\begin{equation}
  \label{eq:expr}
  e^{-t}=\pi^{-1/2}\int_0^{\infty} p^{-3/2}e^{-1/p} e^{-pt^2}dp
\end{equation}
which is a continuum type integral.

The Borel sum gives consistent results, and remaining
exponential terms are uniquely defined. 

Ramified
analyticity of $\hat \psi$ follows from the formulas in \cite{Newton} and \cite{Calderon}. However, the techniques \cite{Calderon} are
more involved and,  along those lines,  there appear to be significant gaps in estimates leading
to a mathematical
proof of Borel summability (which, in fact,  is not the intention
of those works). The purpose of the analysis
in \cite{Calderon}, \cite{Newton}   and related literature is different: the extension of a spectral-like theory and a ``spectral calculus'' beyond the continuous spectrum.

}
\end{Note}

\subsection{Analytic potentials}\label{Ap} Since the wave function is the solution of a PDE which mixes space and time information, finding the detailed
time behavior of $\psi$ is contingent on
detailed information about $V(x)$ and $\psi_0(x)$. It seems likely
that generalized (multi-) summability of the large time
(trans-)series of $\psi$ should hold whenever $V$ has a
multisummable transseries as well. In this paper though we consider
potentials analytic at infinity
 and with sufficient decay.  For simplicity, we
write $-ip=\epsilon^2$. Equation (\ref{eq:hom}) reads:
\begin{equation}
  \label{eq:hom1}
   y''(x)=\left(V(x)+\epsilon^2\right)y(x)
\end{equation}
From the form of the Green's function, it is clear that the analytic
properties of the Green's function at $p=0$ follow from those of the
Jost functions (defined as in (\ref{eq:defypm}):  $y^+$
is the solution that behaves like $e^{-\sqrt{-ip}x}$ as $x\to\infty$,
when $p\in\RR^+$). We analyze potentials of the form $V(x)=a/x^m,
m\in\NN$, and we discuss how essentially the same arguments would
work for any potential which is analytic at infinity and
$O(x^{-m})$. The value $m=2$ marks in a sense a threshold, making
the transition between convergent and divergent expansions in energy
at the bottom of the continuous spectrum. For $m=1, 2$ the equation
can be solved in terms of simple special functions; the slow decay
in the case $m=1$ implies that zero is an accumulation point of
poles; no convergent Frobenius expansion is possible.
\begin{Proposition}
  For $m\ge 2$ the Jost functions have convergent Frobenius expansions
  in $\epsilon$ (series in ramified powers of $\epsilon$ and
  $\epsilon\log\epsilon$ and are bounded by $const e^{|\epsilon ||x|}$
  uniformly in a sector $\arg\epsilon\in(-3\pi/4,3\pi/4)$.
\end{Proposition}
For $V(x)=a/x^2$, (\ref{eq:hom1}) is solved by Bessel functions; the
solution that decays like $e^{-\epsilon x}$ as $x\to \infty$
is given in terms of the Bessel function $K$ as
$$y^+=\sqrt{2\epsilon x/\pi}K_{\alpha}(\epsilon x);\ \alpha =\sqrt{a+1/4}$$
For small $\epsilon$ and fixed $x$, $y^+$ has the form
\begin{equation}
  \label{loc}
  C_1(x) \epsilon^{1+\alpha}A_1(\epsilon)+
 C_2(x) \epsilon^{1-\alpha}A_2(\epsilon)
\end{equation}
with $A_1,A_2$ analytic.
For $m\ge 3$, there are no ramified powers of $\epsilon$ in the expansions,
but all powers of $\epsilon\log\epsilon$ intervene.
\begin{Proposition}
  For fixed $x$ and $m\ge 3$, the function $s(x;\epsilon)$ is of the form $G(\epsilon,\epsilon\log\epsilon)$ where $G(u,v)$ is analytic for small $(u,v)$. The Jost functions
are bounded by $const e^{|\epsilon ||x|}$ uniformly in a sector
$\arg\epsilon\in(-3\pi/4,3\pi/4)$.
\end{Proposition}
\begin{proof}
  The question is the dependence of the Jost function in $\epsilon$,
  for small $\epsilon\in\CC$. As a mathematical question, this is a
  connection problem: the definition of the Jost function is given in
  terms of the asymptotic behavior as $x\to\infty$ while the analyticity properties in
  $\epsilon$ are sought globally in $x$.

  It is convenient to treat this problem by Borel summation once
  again, this time in $x$, to transform it into a pure analyticity
  question. We analyze the Jost function given, for $\epsilon>0$, by
  $y(x)\sim e^{-\epsilon x}(1+s(x;\epsilon))$ where $s(x;\epsilon)$ is
  an $o(1)$ power series in $1/x$, as $x\to\infty$. It is easy to see
  that such a solution (whose existence is known and also follows from
  the argument below) is unique.

For Borel summability, we have to extract $s$.  We thus write $y(x;\epsilon)= e^{-\epsilon x}(1+s(x;\epsilon))$ and obtain
\begin{equation}
  \label{eq:eqs}
  s''-2\epsilon s'-V(x)s=V(x)
\end{equation}
 To simplify even further the presentation we take $m=3$, but there is nothing special
about this choice, and the extension to other values of $m$ is immediate.

We inverse Laplace transform (\ref{eq:eqs})  (the legitimacy of which is justified ``backwards''
by showing that the Laplace transform of the solution of the thus obtained
equation solves (\ref{eq:eqs}), which has a unique small solution) and obtain
\begin{equation}
  \label{eq:eqL}
  H(q)=\frac{1}{2}p^2+\frac{1}{2}\mathcal{P}^3\frac{H(q)}{q(q+2\epsilon)}
\end{equation} where $\mathcal{P}F$ is the definite antiderivative
of $F$ which is zero at zero.  With the change of variable
${\tau}=\epsilon q$, $H(q)=F({\tau})$, we obtain
\begin{equation}
  \label{eq:eq7}
  F({\tau})=\frac{\epsilon^2 {\tau}^2} {2}+\epsilon \mathcal{P}^3 \frac{F({\tau})}{{\tau}({\tau}+2)}
\end{equation}
We look for a solution which are $O(\epsilon^2 {\tau}^2)$ for small ${\tau}$. Consider
the space $\mathcal{B}$ of functions of the form $f({\tau})={\tau}^2 G({\tau})$ where $G$ is analytic
for, say,  $|{\tau}|<1$ with the norm  $\|f\|=\sup_{|{\tau}|<1}|G({\tau})|$. We see
that this is a Banach space, and eq. (\ref{eq:eqL}) is contractive
in $\mathcal{B}$. It is also unique in the space of functions of the form
${\tau}^2 G({\tau})$ with $G$ defined in  $L^1[0,1]$.  The solution
of (\ref{eq:eqL}) is unique, and  analytic for small ${\tau}$.
As a differential equation this reads
\begin{equation}
  \label{eq:eqdiff}
  F'''=\frac{\epsilon F}{{\tau}({\tau}+2)}
\end{equation}
The argument above, or Frobenius theory, shows that (\ref{eq:eqdiff}) also has
a unique solution which is of the form $\frac{1}{2}\epsilon^2 {\tau}^2(1+o(1))$ for small ${\tau}$. The solution is obviously analytic for
$\Re\, {\tau}>-2$, since there the only singularity of the equation is ${\tau}=0$.

By standard ODE asymptotic results \cite{Wasow} we see that any
solution of (\ref{eq:eqdiff}) is uniformly bounded in $\CC$ by
\begin{equation}
  \label{eq:eqC2}
 C |{\tau}|^{2/3}e^{3|{\tau}|^{1/3}}
\end{equation}
for some $C$. This ensures the necessary
(sub)exponential bounds for taking the Laplace transform.

On the other hand, we look for solutions of (\ref{eq:eqdiff}) in the
form
\begin{equation}
  \label{eq:ser2}
  F=\epsilon^2 F_2+\sum_{j\ge 3}\epsilon^j F_j({\tau})
\end{equation}
 The functions $F_j$ satisfy the recurrence
\begin{equation}
  \label{eq:recd}
 F_{j+1}'''=\frac{F_j}{{\tau}({\tau}+2)}, j\ge 3
\end{equation}
With our initial condition, we get $F_2={\tau}^2/2$ and
\begin{equation}
  \label{eq:reci}
 F_{j+1}=\mathcal{P}^3  \frac{F_j}{{\tau}({\tau}+2)}, j\ge 3
\end{equation}
For now, we take ${\tau}$ in the right half plane, $\mathbb{H}$.
 It can be checked by induction that  $F_j$ are analytic in $\mathbb{H}$
and at zero, and
\begin{equation}
  \label{eq:bdF}
|F_j|\le 4 \frac{{\tau}^j}{j!^3}
\end{equation}
It follows that the series (\ref{eq:ser2}) converges uniformly
on any compact set in $\mathbb{H}$. Moreover, we see
that
the function series
\begin{equation}
  \label{eq:ser2}
 H(q)=\frac{q^2}{2}+\sum_{j\ge 3}\epsilon^j F_j(q/\epsilon)
\end{equation}
also converges uniformly in $q$
on any compact set in $\mathbb{H}$.
The Laplace transform of $H$ reads
\begin{equation}
  \label{eq:eqH2}
  \int_0^{\infty}e^{-qx}H(q)dq=\frac{1}{x^3}+\sum_{j\ge 3}\epsilon^j \int_0^{\infty}e^{-qx} F_j(q/\epsilon)dq=:\frac{1}{x^3}+\sum_{j\ge 3}f_j(x;\epsilon)
\end{equation}
where, once more, the interchange of summation and integration
is justified by the bound (\ref{eq:bdF}). We fix $x$, drop
it from the notations, and note that
in the last sum we have $|f_j|\le const (j!)^{-2}$. We thus
need to study the analyticity of $f_j$. We claim that  $f_j(\epsilon)=G_j(\epsilon,\epsilon\log\epsilon)$ where $G_j(u,v)$
is analytic in small $(u,v)$.  Dominated convergence ensures that the integral
on the left side of (\ref{eq:eqH2}) is of the same form.

We will use the following Lemma which applies at the other end of Watson's Lemma setting.
\begin{Lemma}\label{InvW}
  Assume $F$ is bounded on $0,M$ and analytic in a sector $\{z:|z|>R;\arg(z)\in (a,b)\}$
and  $f$ is of the form $z^NG(z^{-1},z^{-1}\log z)$ where
$G(u,v)$ is analytic in the polydisk $\{(u,v):|u|<M_1^{-1}, |v|<M_1^{-1}\}$,
let $M_1>M$ and consider
$$f(s)=\int_0^{\infty}e^{-sp}F(p)dp$$ Then
\begin{equation}
  \label{eq:eqf5}
  f(s)=s^{-N-1}H(s,s\log s)
\end{equation}
where $H(u,v)$ is analytic for small $u,v$. (In a very similar way,
the lemma could accommodate for fractional powers of $z^{-1}$.)
\end{Lemma}
Note that the convergence of the series in $1/x$, $x^{-1}\log x$ entails
that $F$ extends analytically on the Riemann surface of the log in
a neighborhood of infinity.
\begin{proof}
The proof is elementary.   Let $M>M_1$, and first note that $\int_0^{M_1} e^{-sp}F(s)ds$ is entire, and we only need to consider the integral
from $M$ to infinity. Since for some $C>0$ and all $(l,j)>(0,0)$ we have
  \begin{equation}
    \label{eq:34}
    \int_{M}^{\infty} e^{-sp}p^{N-n-l}\log^j(p)dp\le const M^{N-n-l}\log^j(N)\frac{1}{s}
  \end{equation}
the series
\begin{equation}
  \label{eq:F5}
  F(p)=p^N\sum_{k,l}c_{kl} x^{-k} (x^{-1}\log x)^l
\end{equation}
can be integrated term by term and uniform convergence easily entails that
it is enough to show the property for a single term, of the form
\begin{equation}
  \label{eq:oneterm}
  Q(s)=\int_{M}^{\infty} e^{-sp}p^{N-n-l}\log^j(p)dp
\end{equation}
A finite number of integrations by parts, multiplications by $s$
and differentiations in $s$ brings (\ref{eq:F5}) to
\begin{equation}
  \label{eq:F6}
  \int_{M}^{\infty} e^{-sp}dp =e^{-Ms}/s=\frac{1}{s} +entire(s)
\end{equation}
Undoing the operations above on the last expression in (\ref{eq:F6}) easily
completes the proof.
\end{proof}
For large enough $a$ we now write
\begin{equation}
  \label{eq:eqf2}
  f_j=\epsilon^{j+1} \int_0^\infty e^{-sx\epsilon } F_j(s)ds=
\epsilon^{j+1} \int_0^M e^{-sx\epsilon } F_j(s)ds+\epsilon^{j+1} \int_M^{\infty} e^{-sx\epsilon } F_j(s)ds
\end{equation}
where we choose $M$ large enough.

The first integral is manifestly entire in $\epsilon$.  For the second term, we have the following.
\begin{Lemma}
   $F_j(s)=s^jW_j(s^{-1},s^{-1}\log s)$
where $W_j(u,v)$ is analytic for small $u,v$.
\end{Lemma}
\begin{proof}
  Induction from (\ref{eq:reci}): The right side operations on the right
side consist in multiplication by ${\tau}^{-1}({\tau}+1)^{-1}$, and three
definite antiderivatives (from zero).  It is sufficient to show that
each of these operations preserves the structure above, aside from the
leading order behavior which follows from straightforward power
counting. Multiplication by ${\tau}^{-1}({\tau}+1)^{-1}$ clearly preserves the structure
mentioned.

\begin{equation}
  \label{eq:eq5}
  \int_0^{\tau}=\int_0^M+\int_M^{\tau}
\end{equation}
where the first integral is a mere constant, and $M$ is chosen
so that $W_j$ is analytic for $|u|<1/M$ and $|v|<1/M$. We then write
$F_j=s^j\sum_{k\le j+1,l} c_{j;kl}u^kv^l+F_{j1}$ with $u=s^{-1}$, $v=s^{-1}\log s$
and where we see that $F_{j1}=O(s^{-2}\log s^k)$. The sum contains finitely
many terms, and for it the structure follows by explicit integration.
For the second we write
\begin{equation}
  \label{eq:eq5}
  \int_M^{\tau}=\int_M^{\infty}-\int_{\tau}^\infty
\end{equation}
where the first integral on the right is a constant. For the
second one, the structure for large $t$ follows from
term by term integration and straightforward estimates.
\end{proof}
The rest of the proof follows from Lemma \ref{InvW}, noting that, for
fixed $l$, there are only finitely many terms in the expansion at
infinity of $f$ for which the total power of ${\tau}$ exceeds $-l$.

  It is clear that all the arguments above go through if $x^{-3}$
is replaced by $x^{-m},m>3$, except for (\ref{eq:eqC2}) where
the exponent will be $p^{1/m}$ and the power of the prefactor changes.
 The bounds for the Jost functions follow immediately
from  the Laplace representation of $s(x)$ and contour
deformation.

If the potential is analytic and $O(x^{-m})$ at infinity, then the
function $H$ will, in general, have exponential order one, rather than
fractional, and the bounds (\ref{eq:bdF}) are ``worse'', the power of
the factorial becoming one. This can be shown similarly, using
a roughly similar recurrence. In the $O(x^{-3})$ case, one would
get recurrence of the form
\begin{equation}
  \label{eq:reci2}
 F_{j+1}=\mathcal{P}^3  \frac{F_j}{{\tau}({\tau}+2)}+a_1\mathcal{P}^4  \frac{F_{j-1}}{{\tau}({\tau}+2)}+\cdots+a_{j-2}\mathcal{P}^{j+1}  \frac{F_2}{{\tau}({\tau}+2)}, j\ge 3
\end{equation}
where $a_j$ grow at most geometrically. The rest of the proof is roughly the same, but the details are more cumbersome.
\end{proof}
\section{Acknowledgments.}
This work was supported in part by the National Science Foundation
 DMS-0601226 and  DMS-0600369. We are grateful to S.
Goldstein, J. Lebowitz, R. Tumulka and  J. Lukkarinen for very useful
comments and suggestions.

\end{document}